\documentclass{llncs}

\usepackage[english]{babel}
\usepackage{fancyvrb}
\usepackage{paralist}
\usepackage{tikz}
\usetikzlibrary{patterns}
\usepackage{titletoc}
\usepackage{url}
\usepackage{booktabs}
\usepackage{float}

\usepackage{amsmath}
\usepackage{amssymb}

\usepackage{color}
\usepackage{todonotes} 

\newcommand{\shortcite}[1]{\cite{#1}}

\newcommand{\BspPoint}[1]{\filldraw #1 circle(0.04cm)}
\newcommand{\TargetPoint}[1]{\draw #1 circle(0.04cm)}

\newtheorem{Theorem}{Theorem}
\newtheorem{AppendixTheorem}{Theorem}

\newcommand{\Conj}{\textsc{Conj}}

\DeclareMathOperator{\qry}{\mathit{?-}}
\DeclareMathOperator{\qryP}{\qry\limits_{p}}

\newcommand{\pcase}[1]{\paragraph{Case \texttt{#1}}}

\newcommand{\true}{\mathit{true}}
\newcommand{\fail}{\mathit{fail}}
\newcommand{\conj}{\mathit{conj}}

\newcommand{\failure}{\mathit{failure}}

\newcommand{\authcount}[1]{} 

\begin{document}

\title{Disjunctive Delimited Control}

\author{Alexander Vandenbroucke\inst{1} \and Tom Schrijvers\inst{2}}
\authorrunning{A. Vandenbroucke \and T. Schrijvers}
\institute{
  Standard Chartered \email{alexander.vandenbroucke@sc.com}
  \and
  KU Leuven \email{tom.schrijvers@kuleuven.be}
}

\label{firstpage}

\maketitle

\begin{abstract}
  Delimited control is a powerful mechanism for programming language extension
  which has been recently proposed for Prolog (and implemented in SWI-Prolog).
  By manipulating the control flow of a program from inside the language, it
  enables the implementation of powerful features, such as tabling, without
  modifying the internals of the Prolog engine.
  However, its current formulation is inadequate: it does not
  capture Prolog's unique non-deterministic nature which allows multiple ways to
  satisfy a goal.

  This paper fully embraces Prolog's non-determinism with a novel interface for
  \emph{disjunctive} delimited control, which gives the programmer not only control over
  the sequential (conjunctive) control flow, but also over the non-deterministic
  control flow. 
  We provide a meta-interpreter that conservatively extends Prolog with delimited
  control and show that it enables a range of typical Prolog features and
  extensions, now at the library level: findall, cut, branch-and-bound
  optimisation, probabilistic programming, \ldots
\end{abstract}


\section{Introduction}

Delimited control is a powerful programming language mechanism for control flow
manipulation that was developed in the late '80s in the context of functional
programming~\cite{Felleisen:1988,abstracting_control}. Schrijvers et
al.~\shortcite{iclp2013} have recently ported this mechanism to Prolog.

This port enables powerful applications in Prolog, such as high-level
implementations of both tabling~\cite{iclp2015} and algebraic effects \&
handlers~\cite{iclp2016}.  Yet, at the same time, their work leaves much untapped
potential, as it fails to recognise the unique nature of Prolog when compared to
functional and imperative languages that have previously adopted delimited
control.

Indeed, computations in other languages have only one \emph{continuation},
i.e., one way to proceed from the current point to a result. In contrast, at
any point in a Prolog continuation, there may be multiple ways to proceed and
obtain a result. More specifically, we can distinguish 1) the success or
\emph{conjunctive} continuation which proceeds with the current state of the
continuation; and 2) the failure or \emph{disjunctive} continuation which
bundles the alternative ways to proceed, e.g., if the conjunctive continuation
fails.

The original delimited control only accounts for one continuation, which
Schrijvers et al. have unified with Prolog's conjunctive continuation. More
specifically, for a given subcomputation, they allow to wrest the current
conjunctive continuation from its track, and to resume it at leisure,
however many times as desired.
Yet, this entirely ignores the disjunctive continuation, which remains
as and where it is.

In this work, we adapt delimited control to embrace the whole of Prolog and
capture both the conjunctive and the disjunctive continuations. This makes it
possible to manipulate Prolog's built-in search for custom search strategies
and enables clean implementations of, e.g., \texttt{findall/3} and
branch-and-bound.
This new version of delimited control has an executable specification in the
form of a meta-interpreter (Section~\ref{sec:meta-interpreter}), that can
run both the above examples, amongst others.

\section{Overview and Motivation}

\subsection{Background: Conjunctive Delimited Control}
\label{sec:bg-conj}

In earlier work, Schrijvers et al.~\shortcite{iclp2013} have introduced a
Prolog-compatible interface for delimited control that consists of two
predicates: \texttt{reset/3} and \texttt{shift/1}.  To paraphrase their
original description, \texttt{reset(Goal,ShiftTerm,Cont)} executes
\texttt{Goal}, and,
\begin{inparaenum}
  \item if \texttt{Goal} fails, \texttt{reset/3} also fails;
  \item if \texttt{Goal} succeeds, then \texttt{reset/3} also succeeds and
    unifies \texttt{Cont} and \texttt{ShiftTerm} with \texttt{0};
  \item if \texttt{Goal} calls \texttt{shift(Term)}, then the execution of
    \texttt{Goal} is suspended and \texttt{reset/3} succeeds immediately, 
    unifying \texttt{ShiftTerm} with \texttt{Term} and \texttt{Cont} with the
    remainder of \texttt{Goal}.
\end{inparaenum}
The \texttt{shift/reset} pair resembles the more familiar
\texttt{catch/throw} predicates, with the following differences:
\texttt{shift/1} does not copy its argument (i.e., it does not refresh the
variables), it does not delete choice points, and 
also communicates the remainder of \texttt{Goal} to \texttt{reset/3}.

\paragraph{Obliviousness to Disjunctions}

This form of delimited control only captures the conjunctive continuation. For
instance \texttt{reset((shift(a),G1),Term,Cont)} captures in \texttt{Cont} goal
\texttt{G1} that appears in conjunction to \texttt{shift(a)}. In a low-level
operational sense this corresponds to delimited control in other (imperative
and functional) languages where the only possible continuation to capture is
the computation that comes sequentially after the shift. 
Thus this approach is
very useful for enabling conventional applications of delimited control in
Prolog.

In functional and imperative languages delimited control can also be
characterised at a more conceptual level as capturing the entire remainder of a
computation. Indeed, in those languages the sequential
continuation coincides with the entire remainder of a computation. Yet, the
existing Prolog approach fails to capture the entire remainder of a goal, as it
only captures the conjunctive continuation and ignores any disjunctions.
This can be illustrated by the \texttt{reset((shift(a),G1;G2),Term,Cont)} which
only captures the conjunctive continuation \texttt{G1} in \texttt{Cont} 
and not the disjunctive continuation \texttt{G2}. In other words, only the
conjunctive part of the goal's remainder is captured.

This is a pity because disjunctions are a key feature of Prolog and many
advanced manipulations of Prolog's control flow involve manipulating those
disjunctions in one way or another.

\subsection{Delimited Continuations with Disjunction}

This paper presents an approach to delimited control for Prolog that is in line
with the conceptual view that the whole remainder of a goal should be captured,
including in particular the disjunctive continuation.

For this purpose we modify the \texttt{reset/3} interface, where
depending on \texttt{Goal},\\
\texttt{reset(Pattern,Goal,Result)} has three possible outcomes:
\begin{enumerate}
\item
If \texttt{Goal} fails, then the \texttt{reset} succeeds and unifies
\texttt{Result} with \texttt{failure}. For instance,
\begin{Verbatim}[frame=single]
  ?- reset(_,fail,Result).
  Result = failure.
\end{Verbatim}
\item
If \texttt{Goal} succeeds, then 
\texttt{Result} is unified with \texttt{success(PatternCopy,}\\
\texttt{DisjCont)} and the \texttt{reset} succeeds.
Here \texttt{DisjCont} is a goal that represents the disjunctive remainder of
\texttt{Goal}. For instance,
\begin{Verbatim}[frame=single]
  ?- reset(X,(X = a; X = b),Result).
  X = a, Result = success(Y,Y = b).
\end{Verbatim}
Observe that, similar to \texttt{findall/3}, the logical variables in
\texttt{DisjCont} have been renamed apart to avoid interference between the
branches of the computation. To be able to identify any variables of interest
after renaming, we provide \texttt{PatternCopy} as a likewise renamed-apart
copy of \texttt{Pattern}. 
\item 
If \texttt{Goal} calls \texttt{shift(Term)}, then the \texttt{reset} succeeds and
\texttt{Result} is unified with \texttt{shift(Term,ConjCont,PatternCopy,DisjCont)}.
This contains in addition to the disjunctive continuation also the conjunctive
continuation. The latter is not renamed apart and can share variables
with \texttt{Pattern} and \texttt{Term}. For instance,
\begin{Verbatim}[frame=single]
  ?- reset(X,(shift(t),X = a; X = b),Result).
  Result = shift(t,X = a, Y, Y = b).
\end{Verbatim}
\end{enumerate}
Note that \texttt{reset(P,G,R)} always succeeds if \texttt{R} is unbound and
never leaves choicepoints.

\paragraph{Encoding \texttt{findall/3}}

Section~\ref{sec:case-studies} presents a few larger applications, but our
encoding of
\texttt{findall/3} with disjunctive delimited control already gives some idea
of the expressive power:
\begin{Verbatim}[frame=single]
findall(Pattern,Goal,List) :-
  reset(Pattern,Goal,Result),
  findall_result(Result,Pattern,List).

findall_result(failure,_,[]).
findall_result(success(PatternCopy,DisjCont),Pattern,List) :-
  List = [Pattern|Tail],
  findall(PatternCopy,DisjCont,Tail).
\end{Verbatim}
This encoding is structured around a \texttt{reset/3} call of the given
\texttt{Goal} followed by a case analysis of the result. Here we assume that
\texttt{shift/1} is not called in \texttt{Goal}, which is a reasonable
assumption for plain \texttt{findall/3}.

\paragraph{Encoding \texttt{!/0}}
Our encoding of the \texttt{!/0} operator illustrates the use of
\texttt{shift/1}:
\begin{Verbatim}[frame=single]
cut :- shift(cut).  

scope(Goal) :-
  copy_term(Goal,Copy),
  reset(Copy,Copy,Result),
  scope_result(Result,Goal,Copy).

scope_result(failure,_,_) :- 
  fail.
scope_result(success(DisjCopy,DisjGoal),Goal,Copy) :- 
  Goal = Copy.
scope_result(success(DisjCopy,DisjGoal),Goal,Copy) :- 
  DisjCopy = Goal,
  scope(DisjGoal).
scope_result(shift(cut,ConjGoal,DisjCopy,DisjGoal),Goal,Copy) :- 
  Copy = Goal,
  scope(ConjGoal).
\end{Verbatim}
The encoding provides \texttt{cut/0} as a substitute for \texttt{!/0}. Where
the scope of regular cut is determined lexically, we use \texttt{scope/1} here
to define it dynamically. For instance, we encode
\begin{center}
\begin{tabular}{ccc}
\begin{minipage}{0.47\textwidth}
\begin{Verbatim}[frame=single]
p(X,Y) :- q(X), !, r(Y).              
p(4,2).

\end{Verbatim}
\end{minipage}
&
as
&
\begin{minipage}{0.47\textwidth}
\begin{Verbatim}[frame=single]
p(X,Y) :- scope(p_aux(X,Y)).
p_aux(X,Y) :- q(X), cut, r(Y).           
p_aux(4,2).
\end{Verbatim}
\end{minipage}
\end{tabular}
\end{center}
The logic of cut is captured in the definition of \texttt{scope/1}; all the
\texttt{cut/0} predicate does is request the execution of a cut with
\texttt{shift/1}.

In \texttt{scope/1}, the \texttt{Goal} is copied to avoid instantiation by any
of the branches. The copied goal is executed inside a \texttt{reset/3} with the
copied goal itself as the pattern. The \texttt{scope\_result/3} predicate
handles the result: 
\begin{itemize}
\item \texttt{failure} propagates with \texttt{fail};
\item \texttt{success} creates a disjunction to either unify the initial goal with the now instantiated copy
      to propagate bindings, or to invoke the disjunctive continuation;
\item \texttt{shift(cut)} discards the disjunctive continuation and proceeds
      with the conjunctive continuation only.
\end{itemize}

\section{Meta-Interpreter Semantics}
\label{sec:meta-interpreter}

We provide an accessible definition of disjunctive delimited control in
the form of a meta-interpreter.
Broadly speaking, it consists of two parts: the core interpreter, and a top
level predicate to initialise the core and interpret the results.

\subsection{Core Interpreter} 
\begin{figure}[htb!]
\begin{Verbatim}[frame=single,numbers=left]
eval([],PatIn,Disj,PatOut,Result) :- !,
    PatOut = PatIn,
    Disj       = alt(BranchPatIn,BranchGoal),
    Result     = success(BranchPatIn,BranchGoal).
eval([true|Conj],PatIn,Disj,PatOut,Result) :- !,
    eval(Conj,PatIn,Disj,PatOut,Result).
eval([(G1,G2)|Conj],PatIn,Disj,PatOut,Result) :- !,
    eval([G1,G2|Conj],PatIn,Disj,PatOut,Result).
eval([fail|_Conj],_,Disj,PatOut,Result) :- !,
    backtrack(Disj,PatOut,Result).
eval([(G1;G2)|Conj],PatIn,Disj,PatOut,Result) :- !,
    copy_term(alt(PatIn,conj([G2|Conj])),Branch),
    disjoin(Branch,Disj,NewDisj),
    eval([G1|Conj],PatIn,NewDisj,PatOut,Result).
eval([conj(Cs)|Conj],PatIn,Disj,PatOut,Result) :- !,
    append(Cs,Conj,NewConj),
    eval(NewConj,PatIn,Disj,PatOut,Result).
eval([shift(Term)|Conj],PatIn,Disj,PatOut,Result) :- !,
    PatOut = PatIn,
    Disj       = alt(BranchPatIn,Branch),
    Result     = shift(Term,conj(Conj),BranchPatIn,Branch).
eval([reset(RPattern,RGoal,RResult)|Conj],PatIn,Disj,PatOut,Result):- !,
    copy_term(RPattern-RGoal,RPatIn-RGoalCopy),
    empty_alt(RDisj),
    eval([RGoalCopy],RPatIn,RDisj,RPatOut,RResultFresh),
    eval([RPattern=RPatOut,RResult=RResultFresh|Conj]
         ,PatIn,Disj,PatOut,Result).
eval([Call|Conj],PatIn,Disj,PatOut,Result) :- !,
    findall(Call-Body,clause(Call,Body), Clauses),
    ( Clauses = [] -> backtrack(Disj,PatOut,Result)
    ; disjoin_clauses(Call,Clauses,ClausesDisj),
      eval([ClausesDisj|Conj],PatIn,Disj,PatOut,Result)
    ).
\end{Verbatim}
\caption{Meta-Interpreter Core}\label{fig:eval}
\end{figure}

Figure~\ref{fig:eval} defines the interpreter's core predicate,
\texttt{eval(Conj, PatIn, Disj,\\PatOut, Result)}. 
It captures the behaviour of \texttt{reset(Pattern,Goal,Result)} where the goal
is given in the form of a list of goals, \texttt{Conj}, together
with the alternative branches, \texttt{Disj}. The latter is
renamed apart from \texttt{Conj} to avoid conflicting
instantiations.

The pattern that identifies the variables of interest (similar to
\texttt{findall/3}) is present in three forms. Firstly, \texttt{PatIn} is
an input argument that shares the variables of interest with \texttt{Conj} (but
not with \texttt{Disj}). Secondly, \texttt{PatOut} outputs the instantiated
pattern when the goal succeeds or suspends on a \texttt{shift/1}. Thirdly, the
alternative branches \texttt{Disj} are of the form
\texttt{alt(BranchPatIn,BranchGoal)} with their own copy of the pattern.

When the conjunction is empty (1--4), the output pattern is unified with the
input pattern, and \texttt{success/2} is populated with the information from
the alternative branches.

When the first conjunct is \texttt{true/0} (5--6), it is dropped and the meta-interpreter
proceeds with the remainder of the conjunction. When it is a
composite conjunction \texttt{(G1,G2)} (7--8), the individual components are added separately to the
list of conjunctions. 

When the first conjunct is \texttt{fail/0} (9--10), the meta-interpreter
backtracks explicitly by means of auxiliary predicate \texttt{backtrack/3}.
\begin{Verbatim}[frame=single]
backtrack(Disj,PatOut,Result) :-
    ( empty_alt(Disj) ->
       Result = failure
    ; Disj = alt(BranchPatIn,BranchGoal) ->
       empty_alt(EmptyDisj),
       eval([BranchGoal],BranchPatIn,EmptyDisj,PatOut,Result)
    ).

empty_alt(alt(_,fail)).
\end{Verbatim}
If there is no alternative branch, it sets the \texttt{Result} to
\texttt{failure}.
Otherwise, it resumes with the alternative branch.
Note that by managing its own backtracking, \texttt{eval/5} is entirely
deterministic with respect to the meta-level Prolog system.

When the first conjunct is a disjunction \texttt{(G1;G2)} (11--14), the
meta-interpreter adds (a renamed apart copy of) \texttt{(G2,Conj)} to the
alternative branches with \texttt{disjoin/3} and proceeds with
\texttt{[G1|Conj]}.
\begin{Verbatim}[frame=single]
disjoin(alt(_,fail),Disjunction,Disjunction) :- !.
disjoin(Disjunction,alt(_,fail),Disjunction) :- !.
disjoin(alt(P1,G1),alt(P2,G2),Disjunction) :-
    Disjunction = alt(P3, (P1 = P3, G1 ; P2 = P3, G2)).
\end{Verbatim}
Note that we have introduced a custom built-in \texttt{conj(Conj)}  that turns
a list of goals into an actual conjunction. It is handled (15--17) by
prepending the goals to the current list of conjuncts, and never actually builds
the explicit conjunction.

When the first goal is \texttt{shift(Term)} (18--21), this is handled similarly
to an empty conjunction, except that the result is a \texttt{shift/4} term
which contains \texttt{Term} and the remainder of the conjunction in addition
the branch information.

When the first goal is a \texttt{reset(RPattern,RGoal,RResult)} (22--27), the
meta-interpreter sets up an isolated call to \texttt{eval/5} for this goal.
When the call returns, the meta-interpreter passes on the results and resumes
the current conjunction \texttt{Conj}. Notice that we are careful
that this does not result in meta-level failure by meta-interpreting the
unification.

Finally, when the first goal is a call to a user-defined predicate (28--33), the
meta-interpreter collects the bodies of the predicate's clauses whose head
unifies with the call. If there are none, it backtracks explicitly. Otherwise,
it builds an explicit disjunction with \texttt{disjoin\_clauses}, which it
pushes on the conjunction stack.
\begin{Verbatim}[frame=single]
disjoin_clauses(_G,[],fail) :- !.
disjoin_clauses(G,[GC-Clause],(G=GC,Clause)) :- !.
disjoin_clauses(G,[GC-Clause|Clauses], ((G=GC,Clause) ; Disj)) :-
    disjoin_clauses(G,Clauses,Disj).
\end{Verbatim}
An example execution trace of the interpreter can be found in
\ref{sec:example-trace}.

\paragraph{Toplevel} The \texttt{toplevel(Goal)}-predicate initialises
the core interpreter with a conjunction containing only the given goal, the
pattern and pattern copy set to (distinct) copies of the goal, and an empty
disjunction.  It interprets the result by non-deterministically producing all
the answers to \texttt{Goal} and signalling an error for any unhandled
\texttt{shift/1}.
\begin{Verbatim}[frame=single]
toplevel(Goal) :-
    copy_term(Goal,GoalCopy),
    PatIn = GoalCopy,
    empty_alt(Disj),
    eval([GoalCopy],PatIn,Disj,PatOut,Result),
    ( Result = success(BranchPatIn,Branch) ->
        ( Goal = PatOut ; Goal = BranchPatIn, toplevel(Branch))
    ; Result = shift(_,_,_,_) ->
        write('toplevel: uncaught shift/1.\n'), fail
    ; Result = failure ->
        fail
    ).
\end{Verbatim}

\section{Case Studies}
\label{sec:case-studies}

To illustrate the usefulness and practicality of our approach, we present two
case studies that use the new \texttt{reset/3} and \texttt{shift/1}.

\subsection{Branch-and-Bound: Nearest Neighbour Search}

\begin{figure}[tb!]
  \centering
\begin{minipage}{0.95\textwidth}  
\begin{Verbatim}[frame=single]
bound(V) :- shift(V).
  
bb(Value,Data,Goal,Min) :-
    reset(Data,Goal,Result),
    bb_result(Result,Value,Data,Min).

bb_result(success(BranchCopy,Branch),Value,Data,Min) :-
  ( Data @< Value -> bb(Data,BranchCopy,Branch,Min)
  ; bb(Value,BranchCopy,Branch,Min)
  ).
bb_result(shift(ShiftTerm,Cont,BranchCopy,Branch),Value,Data,Min) :-
  (  ShiftTerm @< Value ->
     bb(Value,Data,(Cont ; (BranchCopy = Data,Branch)),Min)
  ;  bb(Value,BranchCopy,Branch,Min)
  ).
bb_result(failure,Value,_Data,Min) :- Value = Min.
\end{Verbatim}
\end{minipage}
\caption{Branch-and-Bound Effect Handler.}
\label{fig:branch-and-bound}
\end{figure}

Branch-and-bound is a well-known general optimisation strategy, where the
solutions in certain areas or branches of the search space are known to be
bounded.
Such branches can be pruned, when their bound does not improve upon a
previously found solution, eliminating large swaths of the search space in a
single stroke.

We provide an implementation of branch-and-bound (see
Figure~\ref{fig:branch-and-bound}) that is generic, i.e., it is not specialised
for any application.
In particular it is not specific to nearest neighbour search, the problem
on which we demonstrate the branch-and-bound approach here.

The framework requires minimal instrumentation: it suffices to begin every
prunable branch with \texttt{bound(V)}, where \texttt{V} is a lower
bound on the values in the branch.\footnote{The framework searches for a
  minimal solution.}
\begin{enumerate}
\item If the \texttt{Goal} succeeds normally (i.e., \texttt{Result} is
  \texttt{success}), then \texttt{Data} contains a new solution, which is only
  accepted if it is an improvement over the existing \texttt{Value}.
  The handler then tries the next \texttt{Branch}.
\item If the \texttt{Goal} calls \texttt{bound(V)}, \texttt{V} is compared
  to the current best \texttt{Value}:
  \begin{itemize}
  \item if it is less than the current value, then \texttt{Cont} could
    produce a solution that improves upon the current value, and thus must
    be explored.
    The alternative \texttt{Branch} is disjoined to \texttt{Cont}, and
    \texttt{DataCopy} is restored to \texttt{Data} (ensuring that a future
    \texttt{reset/3} copies the right variables);
  \item if it is larger than or equal to the current value, then \texttt{Cont}
    can be safely discarded.
  \end{itemize}
\item Finally, if the goal fails entirely, \texttt{Min} is the current
  minimum \texttt{Value}.

\end{enumerate}


\begin{figure}[tb!]
\begin{minipage}{\textwidth}
\begin{Verbatim}[frame=single]
nn((X,Y),BSP,D-(NX,NY)) :-
    ( BSP = xsplit((SX,SY),Left,Right) ->
        DX is X - SX, 
        branch((X,Y), (SX,SY), Left, Right, DX, D-(NX,NY))
    ; BSP = ysplit((SX,SY),Up,Down) ->
        DY is Y - SY, 
        branch((X,Y), (SX,SY), Up, Down, DY, D-(NX,NY))
    ).
branch((X,Y), (SX,SY), BSP1, BSP2, D, Dist-(NX,NY)) :-
    ( D < 0 -> % Find out which partition contains (X,Y).
        TargetPart = BSP1, OtherPart = BSP2, BoundaryDistance is -D
    ;  
        TargetPart = BSP2, OtherPart = BSP1, BoundaryDistance is D
    ),
    ( nn((X,Y), TargetPart, Dist-(NX,NY))
    ; Dist is (X - SX) * (X - SX) + (Y - SY) * (Y - SY),
      (NX,NY) = (SX,SY)
    ; bound(BoundaryDistance-nil),
      nn((X,Y), OtherPart,Dist-(NX,NY))
    ).
run_nn((X0,Y0),BSP,(NX,NY)) :-
    toplevel(bb(10-nil,D-(X,Y),nn((X0,Y0),BSP,D-(X,Y)),_-(NX,NY))).
\end{Verbatim}
\end{minipage}
\caption{2D Nearest Neighbour Search with Branch-and-Bound.}
\label{fig:nearest-neighbour}
\end{figure}
\paragraph{Nearest Neighbour Search} The code in
Figure~\ref{fig:nearest-neighbour} shows how the branch and bound
framework efficiently solves the problem of finding the point (in a given
set) that is nearest to a given target point on the Euclidean plane.

The \texttt{run\_nn/3} predicate takes a point \texttt{(X,Y)},
a Binary Space Partitioning (BSP)-tree\footnote{A BSP-tree is a tree that
  recursively partitions a set of points on the Euclidean plane, by picking
  points and alternately splitting the plane along the x- or y-coordinate of
  those point. Splitting along the x-coordinate produces an \texttt{xsplit/3}
  node, along the y-coordinate produces a \texttt{ysplit/3} node.}
that represents the set of points, and returns the point, nearest to
\texttt{(X,Y)}.
The algorithm implemented by \texttt{nn/3} recursively descends the BSP-tree.
At each node it first tries the partition to which the target point belongs,
then the point in the node, and finally the other partition.
For this final step we can give an easy lower bound: any point in the
other partition must be at least as far away as the (perpendicular) distance
from the given point to the partition boundary.

\begin{figure}[tb!]
  \centering
  \begin{tikzpicture}[scale=2.2]
    \draw[pattern=north west lines, pattern color=black,opacity=0.5]
      (-1,-1) -- (0,-1) -- (0,1) -- (-1,1) -- cycle;
    \draw (-1,-1) -- (1,-1) -- (1,1) -- (-1,1) -- cycle;
    \TargetPoint{(1,0.1)}   node[anchor=east]  {(1,0.1)};
    \BspPoint{(0,0)}        node[anchor=west]  {(0,0)};
    \BspPoint{(0.5,0.5)}    node[anchor=south] {(0.5,0.5)};
    \BspPoint{(-0.5,0)}     node[anchor=south] {(-0.5,0.5)};
    \BspPoint{(-0.75,-0.5)} node[anchor=west]  {(-0.75,-0.5)};
    \draw (0,-1) -- (0,1);
    \draw (0,0.5) -- (1,0.5);
    \draw (-1,0) -- (0,0);
    \draw (-0.75,0) -- (-0.75,-1);
  \end{tikzpicture}
  \caption{Nearest-Neighbour Search using a BSP-tree}
  \label{fig:bsp}
\end{figure}
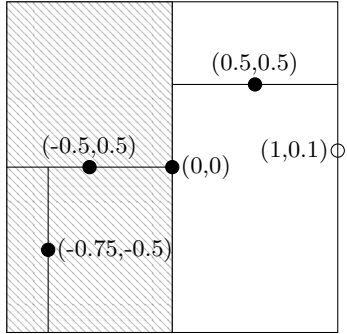

As an example, we search for the point nearest to $(1,0.1)$ in the set
$\{(0.5,0.5),$ $(0,0),$ $(-0.5,0),$ $(-0.75,-0.5)\}$.
Figure~\ref{fig:bsp} shows a BSP-tree containing these points,
the solid lines demarcate the partitions.
The algorithm visits the points $(0.5,0.5)$ and $(0,0)$, in that order.
The shaded area is never visited, since the distance from (1,0.1) to the
vertical boundary through $(0,0)$ is greater than the distance to $(0.5,0.5)$
(1 and about 0.64).
The corresponding call to \texttt{run\_nn/3} is:
\begin{Verbatim}[frame=single]
  ?- BSP = xsplit((0,0),
            ysplit((-0.5,0),leaf,xsplit((-0.75,-0.5),leaf,leaf)),
            ysplit((0.5,0.5),leaf,leaf)),
     run_nn((1,0.1),BSP,(NX,NY)).
  NX = NY, NY = 0.5.
\end{Verbatim}

\subsection{Probabilistic Programming}
Probabilistic programming languages (PPLs) are programming languages 
designed for probabilistic modelling.
In a probabilistic model, components behave in a variety of ways---just like in a
non-deterministic model---but do so with a certain probability.

Instead of a single deterministic value, the execution of a probabilistic
program results in a probability distribution of a set of values.
This result is produced by probabilistic
\emph{inference}~\cite{DBLP:conf/aistats/WoodMM14,DBLP:journals/tplp/FierensBRSGTJR15},
for which there are many strategies
and algorithms, the discussion of which is out of scope here.
Here, we focus on one concrete probabilistic \emph{logic} programming
languages: PRISM~\cite{DBLP:conf/iclp/Sato09}.

A PRISM program consists of Horn clauses, and in fact,
looks just like a regular Prolog program.
However, we distinguish two special predicates:
\vspace{-1mm}
\begin{itemize}
\item
  \texttt{values\_x(Switch,Values,Probabilities)} This predicate defines a
  probabilistic switch \texttt{Switch}, that can assume a value from
  \texttt{Values} with the probability that is given at the corresponding
  position in \texttt{Probabilities} (the contents of \texttt{Probabilities}
  should sum to one).
\item \texttt{msw(Switch,Value)} This predicate samples a value
  \texttt{Value} from a switch \texttt{Switch}.
  For instance, if the program contains a switch declared as
  \texttt{values\_x( coin, [h,t], [0.4,0.6])}, then \texttt{msw(coin,V)}
  assigns \texttt{h} (for heads) to \texttt{V} with probability 0.4, and
  \texttt{t} (for tails) with probability 0.6.
  Remark that each distinct call to \texttt{msw} leads to a different sample
  from that switch.
  For instance, in the query \texttt{msw(coin,X),msw(coin,Y)}, the outcome
  could be either \texttt{(h,h)},\texttt{(t,t)}, \texttt{(h,t)} or
  \texttt{(t,h)}.
\end{itemize}
Consider the following PRISM program, the running example for this section:
\begin{Verbatim}[frame=single]
  values_x(coin1,[h,t],[0.5,0.5]).
  values_x(coin2,[h,t],[0.4,0.6]).
  twoheads :- msw(coin1,h),msw(coin2,h).
  onehead :- msw(coin1,V), (V = t, msw(coin2,h) ; V = h).
\end{Verbatim}
This example defines two predicates: \texttt{twoheads} which is true if both
coins are heads, and \texttt{onehead} which is true if either coin is heads.
However, note the special structure of \texttt{onehead}: PRISM requires the
\emph{exclusiveness condition}, that is, branches of a disjunction cannot be
both satisfied at the same time.
The simpler goal \texttt{msw(coin1,heads) ; msw(coin2, heads)} violates this
assumption.

\begin{figure}[tb!]
\begin{Verbatim}[frame=single,numbers=left]
msw(Key,Value) :- shift(msw(Key,Value)).
prob(Goal) :-
    prob(Goal,ProbOut),
    write(Goal), write(': '), write(ProbOut), write('\n').
prob(Goal,ProbOut) :-
    copy_term(Goal,GoalCopy),
    reset(GoalCopy,GoalCopy,Result),
    analyze_prob(GoalCopy,Result,ProbOut).
analyze_prob(_,failure,0.0).
analyze_prob(_,success(_,_),1.0).
analyze_prob(_,shift(msw(K,V),C,_,Branch),ProbOut) :-
    values_x(K,Values,Probabilities),
    msw_prob(V,C,Values,Probabilities,0.0,ProbOfMsw),
    prob(Branch,BranchProb),
    ProbOut is ProbOfMsw + BranchProb.
\end{Verbatim}
\caption{An implementation of probabilistic programming with delimited
  control.}
\label{fig:prob-prog}
\end{figure}

The code in Figure~\ref{fig:prob-prog} interprets this program.
Line 1 defines \texttt{msw/2} as a simple shift.
Lines 6--9 install a \texttt{reset/3} call over the goal, and analyse the
result.
The result is analysed in the remaining lines:
A \emph{failure} never succeeds, and thus has success probability 0.0 (line 9).
Conversely, a successful computation has a success probability of 1.0 (line 10).
Finally, the probability of a switch (lines 11-15) is the sum of the
probability of the remainder of the program given each possible value of the
switch multiplied with the probability of that value, and summed with the
probability of the alternative branch.

The predicate \texttt{msw\_prob} finds the joint probability of all choices.
It iterates over the list of values, and sums the probability of their
continuations.
\begin{Verbatim}[frame=single]
  msw_prob(_,_,[],[],Acc,Acc).
  msw_prob(V,C,[Value|Values],[Prob|Probs],Acc,ProbOfMsw) :-
    prob((V = Value,C),ProbOut),
    msw_prob(V,C,Values,Probs,Prob*ProbOut + Acc,ProbOfMsw).
\end{Verbatim}

Now, we can compute the probabilities of the two predicates above:
\begin{Verbatim}[frame=single]
  ?- toplevel(prob(twoheads)).
  twoheads: 0.25
  ?- toplevel(prob(onehead)).
  onehead: 0.75
\end{Verbatim}
In Appendix~\ref{sec:problog} we show how to implement the semantics a definite,
non-looping fragment of ProbLog~\cite{DBLP:journals/tplp/FierensBRSGTJR15},
another logic PPL, on top of the code in this section.

\section{Properties of the Meta-Interpreter}
In this section we establish two important correctness properties of our
meta-interpreter.
The proofs of these properties are in the corresponding Appendices
\ref{sec:soundness} and \ref{sec:completeness}
The first theorem establishes the soundness of the meta-interpreter, i.e.,
if a program (not containing \texttt{shift/1} or \texttt{reset/3}) evaluates
to success, then an SLD-derivation of the same answer must exist.

\vspace{5mm}

\begin{Theorem}[Soundness]
  For all lists of goals $[A_1,\ldots,A_n]$,
  terms $\alpha,\beta,\gamma,\nu$, variables $P,R$
  conjunctions $B_1,\ldots,B_m$; $C_1,\ldots,C_k$ and substitutions $\theta$,
  if
  \[
    \begin{array}{l}
      ?- \mathit{eval}([A_1,\ldots,A_n],\alpha,\mathit{alt}(\beta,(B_1,\ldots,B_m)),P,R).\\
      P = \nu, R = \mathit{success(\gamma,C_1,\ldots,C_k)}.
    \end{array}
  \]
  and the program contains neither \texttt{shift/1} nor \texttt{reset/3},
  then SLD-resolution\footnote{Standard SLD-resolution, augmented with
    disjunctions and \texttt{conj/1} goals.}
    finds the following derivation:
  \[
    \begin{array}{c}
      \leftarrow (A_1,\ldots,A_n,\true) ; (\alpha = \beta,B_1,\ldots,B_m)\\
      \vdots\\
      \square\\
      \text{(with solution $\theta$ s.t. $\alpha\theta = \nu$)}
    \end{array}
  \]
\end{Theorem}
Conversely, we want to argue that the meta-interpreter is complete,
i.e., if SLD-derivation finds a refutation, then
meta-interpretation---provided that it terminates---must find the same
answer eventually.
The theorem is complicated somewhat by the fact that the first answer that
the meta interpreter arrives at might not be the desired one due to the
order of the clauses in the program.
To deal with this problem, we use the operator $\qryP$, which is like $\qry$,
but allows a different permutation of the program in every step.
\begin{Theorem}[Completeness]
  For all goals $\leftarrow A_1,\ldots,A_n$, if has solutiono $\theta$, 
  then
  \[\begin{array}{l}
  \qryP \mathit{eval}([A_1,\ldots,A_n],\alpha,\mathit{alt}(\beta,(B_1,\ldots,B_m)),P,R).\\
  P = \mathit{success}(\gamma,(C_1,\ldots,C_k)),R = \alpha\theta.
  \end{array}\]
\end{Theorem}

Together, these two theorems show that our meta-interpreter is a conservative
extension of the conventional Prolog semantics.

\section{Related Work}

\paragraph{Conjunctive Delimited Control}
Disjunctive delimited control is the culmination of a line of research on
mechanisms to modify Prolog's control flow and search, which started with the
hook-based approach of \textsc{Tor}~\cite{DBLP:journals/scp/SchrijversDTD14}
and was followed by the development of conjunctive delimited control
for Prolog~\cite{iclp2013,DBLP:conf/ppdp/SchrijversWDD14}.

The listing below shows that disjunctive delimited control entirely
subsumes conjunctive delimited control. The latter behaviour is recovered
by disjoining the captured disjunctive branch. We believe that \textsc{Tor}
is similarly superseded.

\begin{Verbatim}[frame=single]
  nd_reset(Goal,Ball,Cont) :-
    copy_term(Goal,GoalCopy),
    reset(GoalCopy,GoalCopy,R),
    ( R = failure -> fail
    ; R = success(BranchPattern,Branch) ->
      ( Goal = GoalCopy, Cont = 0
      ; Goal = BranchPattern, nd_reset(Branch,Ball,Cont))
    ; R = shift(X,C,BranchPattern,Branch) ->
      ( Goal = GoalCopy, Ball = X, Cont = C
      ; Goal = BranchPattern, nd_reset(Branch,Ball,Cont))
    ).
\end{Verbatim}

Abdallah~\shortcite{DBLP:journals/corr/abs-1708-07081} presents a higher-level
interface for (conjunctive) delimited control on top of that of Schrijvers et
al.~\shortcite{iclp2013}. In particular, it features \emph{prompts}, first
conceived in a Haskell implementation by Dyvbig et al.~\shortcite{prompts},
which allow shifts to dynamically specify up to what reset to capture the
continuation. We believe that it is not difficult to add a similar prompt
mechanism on top of our disjunctive version of delimited control.

\paragraph{Interoperable Engines}
Tarau and Majumdar's Interoperable Engines~\shortcite{DBLP:conf/padl/TarauM09}
propose \emph{engines} as a means for co-operative coroutines in Prolog.  An
engine is an independent instance of a Prolog interpreter that provides answers
to the main interpreter on request.


The predicate \texttt{new\_engine(Pattern,Goal,Interactor)} creates a new
engine with answer pattern \texttt{Pattern} that will execute \texttt{Goal} and
is identified by \texttt{Interactor}.
The predicate \texttt{get(Interactor,Answer)} has an engine execute its goal
until it produces an answer (either by proving the \texttt{Goal}, or explicitly
with \texttt{return/1}).  After this predicate returns, more answers can be
requested, by calling \texttt{get/2} again with the same engine identifier.
The full interface also allows bi-directional communication between engines, but
that is out of scope here.

\begin{figure}[tb!]
\begin{Verbatim}[frame=single]
  get(Interactor,Answer) :-
    get_engine(Interactor,Engine),        % get engine state
    run_engine(Engine,NewEngine,Answer),  % run up to the next answer
    update_engine(Interactor,NewEngine).  % store the new engine state
  return(X) :- shift(return(X)).
  run_engine(engine(Pattern,Goal),NewEngine,Answer) :-
    reset(Pattern,Goal,Result),
    run_engine_result(Pattern,NewEngine,Answer,Result).
  run_engine_result(Pattern,NewEngine,Answer,failure) :-
    NewEngine = engine(Pattern,fail),
    Answer    = no.
  run_engine_result(Pattern,NewEngine,Answer,success(BPattern,B)) :-
    NewEngine = engine(BPattern,B),
    Answer    = the(Pattern).
  run_engine_result(Pattern,NewEngine,Answer,S) :-
    S = shift(return(X),C,BPattern,B)
    BPattern  = Pattern,
    NewEngine = engine(Pattern,(C;B)),
    Answer    = the(X).
\end{Verbatim}
\caption{Interoperable Engines in terms of delimited control.}
\label{fig:interop-engines}
\end{figure}

Figure~\ref{fig:interop-engines} shows that we can implement the \texttt{get/2}
engine interface in terms of delimited control (the full code is available
in Appendix~\ref{sec:app-interop-engines}).
The opposite, implementing disjunctive delimited control with engines, seems
impossible as engines do not provide explicit control over the disjunctive
continuation. Indeed, \texttt{get/2} can only follow Prolog's natural
left-to-right control flow and thus we cannot, e.g., run the disjunctive
continuation before the conjunctive continuation, which is trivial with
disjunctive delimited control.

\paragraph{Tabling without non-bactrackable variables}
Tabling~\cite{swift2012,yap} is a well-known technique that eliminates the
sensitivity of SLD-resolution to clause and goal ordering, allowing a larger
class of programs to terminate.
As a bonus, it may improve the run-time performance (at the expense of
increased memory consumption).

One way to implement tabling---with minimal engineering impact to the Prolog
engine---is the tabling-as-a-library approach proposed by Desouter et
al.~\shortcite{iclp2015}.
This approach requires (global) mutable variables that are not erased by
backtracking to store their data structures in a persistent manner.
With the new \texttt{reset/3} predicate, this is no longer needed, as
(non-backtracking) state can be implemented in directly with disjunctive
delimited control.

\section{Conclusion and Future Work}
We have presented \emph{disjunctive delimited control}, an extension to
delimited control that takes Prolog's non-deterministic nature into account.
This is a conservative extension that enables implementing disjunction-related
language features and extensions as a library.

In future work, we plan to explore a WAM-level implementation of disjunctive
delimited control, inspired by the stack freezing functionality of tabling
engines, to gain access to the disjunctive continuations efficiently.
Similarily, the use of \texttt{copy\_term/2} necessitated by the current
API has a detrimental impact on performance, which might be overcome by a
sharing or shallow copying scheme.


\bibliographystyle{splncs04}
\bibliography{disj}

\newpage
\appendix
\section*{Appendix Table of Contents}
\startcontents[sections]
\printcontents[sections]{l}{1}{\setcounter{tocdepth}{2}}
\newpage

\section{Correctness Proofs}

\subsection{Evaluation Is Sound}
\label{sec:soundness}
\begin{AppendixTheorem}[Soundness]
  For all lists of goals $[A_1,\ldots,A_n]$,
  terms $\alpha,\beta,\gamma,\nu$, variables $P,R$
  conjunctions $B_1,\ldots,B_m$; $C_1,\ldots,C_k$ and substitutions $\theta$,
  if
  \[
    \begin{array}{l}
      \qry \mathit{eval}([A_1,\ldots,A_n],\alpha,\mathit{alt}(\beta,(B_1,\ldots,B_m)),P,R).\\
      P = \nu, R = \mathit{success(\gamma,C_1,\ldots,C_k)}.
    \end{array}
  \]
  and the program contains neither \texttt{shift/1} nor \texttt{reset/3},
  then SLD-resolution\footnote{Standard SLD-resolution, augmented with
    disjunctions and \texttt{conj/1} goals.}
    to finds the following derivation:
  \[
    \begin{array}{c}
      \leftarrow (A_1,\ldots,A_n,\true) ; (\alpha = \beta,B_1,\ldots,B_m)\\
      \vdots\\
      \square\\
      \text{(with solution $\theta$ s.t. $\alpha\theta = \nu$)}
    \end{array}
  \]
\end{AppendixTheorem}

\begin{proof}
  By induction on the Prolog derivation.
  
  \pcase{[]} In this case the only possible result is $P = \alpha$:
  \begin{align*}
    \begin{array}{l}
      \qry \mathit{eval}([],\alpha,\mathit{alt}(\beta,(B_1,\ldots,B_m)),P,R).\\
      P = \alpha,
      R = \mathit{success(\beta,(B_1,\ldots,B_m))}.
    \end{array}
    \Longrightarrow
    \begin{array}{c}
      \leftarrow \true ; (\alpha = \beta,B_1,\ldots,B_m)\\
      \leftarrow \true (= \square)\\
      \text{(with solution $\epsilon$)}
    \end{array}
  \end{align*}
  where $\epsilon$ is the empty substitution.

  \pcase{[true|\Conj]} In this case the body contains only a direct recursive
  call, we have:
  \begin{align*}
    &
    \begin{array}{l}
      \qry \mathit{eval}([\true,A_2,\ldots,A_n],\alpha,\mathit{alt}(\beta,(B_1,\ldots,B_m)),P,R).\\
      P = \nu,
      R = \mathit{success(\gamma,(C_1,\ldots,C_k))}.
    \end{array}\\
    \iff&
    \begin{array}{l}
      \qry \mathit{eval}([A_2\ldots,A_n],\alpha,\mathit{alt}(\beta,(B_1,\ldots,B_m)),P,R).\\
      P = \nu,
      R = \mathit{success(\gamma,(C_1,\ldots,C_k))}.
    \end{array}
    & \fbox{\parbox[t]{11em}{\textbf{Definition and Determinism of eval/5}}}\\
    \Longrightarrow&
    \begin{array}{c}
      \leftarrow (A_2,\ldots,A_n,\true);(\alpha = \beta,B_1,\ldots,B_m)\\
      \vdots\\
      \square\\
      \text{(with solution $\theta$ s.t. $\alpha\theta = \nu$)}
    \end{array}
    &\fbox{\textbf{Induction}}\\
    \Longrightarrow&
    \begin{array}{c}
      \leftarrow (\true,A_2,\ldots,A_n,\true);(\alpha = \beta,B_1,\ldots,B_m)\\
      \vdots\\
      \square\\
      \text{(with solution $\theta$ s.t. $\alpha\theta = \nu$)}
    \end{array}
    & \fbox{\textbf{SLD-resolution}}
  \end{align*}

  \pcase{[(G1,G2)|\Conj]} A straightforward calculation gives the desired result:
  \begin{align*}
    &
    \begin{array}{l}
      \qry \mathit{eval}([(G_1,G_2),A_2,\ldots,A_n],\alpha,\mathit{alt}(\beta,(B_1,\ldots,B_m)),P,R).\\
      P = \nu,
      R = \mathit{success(\gamma,(C_1,\ldots,C_k))}.
    \end{array}\\
    &\fbox{\textbf{Definition and Determinism of eval/5}}\\
    \iff&
    \begin{array}{l}
      \qry \mathit{eval}([G_1,G_2,A_2\ldots,A_n],\alpha,\mathit{alt}(\beta,(B_1,\ldots,B_m)),P,R).\\
      P = \nu,
      R = \mathit{success(\gamma,(C_1,\ldots,C_k))}.
    \end{array}\\
    \Longrightarrow&
    \begin{array}{c}
      \leftarrow (G_1,G_2,A_2,\ldots,A_n,\true);(\alpha = \beta,B_1,\ldots,B_m)\\
      \vdots\\
      \square\\
      \text{(with solution $\theta$ s.t. $\alpha\theta = \nu$)}
    \end{array}
    &\fbox{\textbf{Induction}}\\
    \Longrightarrow&
    \begin{array}{c}
      \leftarrow ((G_1,G_2),A_2,\ldots,A_n,\true);(\alpha = \beta,B_1,\ldots,B_m)\\
      \vdots\\
      \square\\
      \text{(with solution $\theta$ s.t. $\alpha\theta = \nu$)}
    \end{array}
    & \fbox{\textbf{Associativity}}
  \end{align*}

  \pcase{[fail|\Conj]} This case immediately calls \texttt{backtrack/3}:
  \begin{align*}
    &
    \begin{array}{l}
      \qry \mathit{eval}([\fail,A_2,\ldots,A_n],\alpha,\mathit{alt}(\beta,(B_1,\ldots,B_m)),P,R).\\
      P = \nu,
      R = \mathit{success(\gamma,(C_1,\ldots,C_k))}.
     \end{array}\\
    \iff&
    \begin{array}{l}
      \qry \mathit{backtrack}(\mathit{alt}(\beta,(B_1,\ldots,B_m)),P,R).\\
      P = \nu,
      R = \mathit{success(\gamma,(C_1,\ldots,C_k))}.
    \end{array}
    &\fbox{\parbox[t]{11em}{\textbf{Definition and Determinism of eval/5}}}
  \end{align*}
  At this point, we can see that $\mathit{alt}(\beta,(B_1,\ldots,B_m))$ cannot
  be empty, for otherwise $R = \failure$.
  Execution of the else branch then gives:
  \begin{align*}
    &
    \begin{array}{l}
      \qry \mathit{eval}([B_1,\ldots,B_m],\beta,\mathit{alt}(\delta,\fail),P,R).\\
      P = \nu,
      R = \mathit{success(\gamma,(C_1,\ldots,C_k))}.
    \end{array}
    &\fbox{\textbf{$\delta$ is fresh}}\\
    \Longrightarrow&
    \begin{array}{c}
      \leftarrow (B_1,\ldots,B_m,true) ; (\beta = \delta, \fail)\\
      \vdots\\
      \square\\
      \text{(with solution $\theta$ s.t. $\beta\theta = \nu$)}
    \end{array}
    &\fbox{\textbf{Induction}}\\
    \Longrightarrow&
    \begin{array}{c}
      \leftarrow (B_1,\ldots,B_m,true)\\
      \vdots\\
      \square\\
      \text{(with solution $\theta$ s.t. $\beta\theta = \nu$)}
    \end{array}
    &\fbox{\textbf{$\leftarrow \beta = \delta, \fail$ is irrefutable}}\\
    \Longrightarrow&
    \begin{array}{c}
      \leftarrow (\alpha = \beta,B_1,\ldots,B_m,true)\\
      \vdots\\
      \square\\
      \text{(with solution $\theta$ s.t. $\alpha\theta = \nu$)}
    \end{array}
    &\fbox{\textbf{SLD-resolution}}\\
    \Longrightarrow&
    \begin{array}{c}
      \leftarrow (\fail,A_2,\ldots,A_n,true); (\alpha = \beta,B_1,\ldots,B_m,true)\\
      \vdots\\
      \square\\
      \text{(with solution $\theta$ s.t. $\alpha\theta = \nu$)}
    \end{array}
    &\fbox{\textbf{SLD-resolution}}
  \end{align*}

  \pcase{[(G1;G2),\Conj]} Assume that the result of the \texttt{copy\_term/2}
  call is $\mathit{alt}(\delta,\conj([D_1,\ldots,D_n]))$, where all the variables
  are fresh, such that
  \begin{equation}\label{eq:copy-term}
    \delta\sigma = \alpha\sigma \Rightarrow (D_1,\ldots,D_n)\sigma = (G_2,A_2,\ldots,A_n)\sigma
  \end{equation}
  Then the result of disjoining this with $\mathit{alt}(\beta,(B_1,\ldots,B_m))$
  is the new disjunction
  \[
  \xi = \mathit{alt}(F,((\delta = F,\conj([D_1,\ldots,D_n]));(\beta = F,B_1,\ldots,B_m)))
  \]
  with $F$ fresh.
  Now we reason as follows:
  \begin{align*}
    &
    \begin{array}{l}
      \qry \mathit{eval}([(G_1;G_2)A_2,\ldots,A_n],\alpha,\mathit{alt}(\beta,(B_1,\ldots,B_m)),P,R).\\
      P = \nu,
      R = \mathit{success(\gamma,(C_1,\ldots,C_k))}.
    \end{array}\\
    \iff&
    \begin{array}{l}
      \qry \mathit{eval}([G_1,A_2,\ldots,A_n],\alpha,\xi,P,R).\\
      P = \nu,
      R = \mathit{success(\gamma,(C_1,\ldots,C_k))}.
    \end{array}\\
    &\fbox{\textbf{Definition and Determinism of eval/5}}\\
    \Longrightarrow&
    \begin{array}{c}
      \leftarrow (G_1,A_2,\ldots,A_n,\true) ; (\alpha = F, ((\delta = F,\conj([D_1,\ldots,D_n]));(\beta = F,B_1,\ldots,B_m)))\\
      \vdots\\
      \square\\
      \text{(with solution $\theta$ s.t. $\alpha\theta = \nu$)}
    \end{array}\\
    &\fbox{\textbf{Associativity, Distributivity, resolution of \texttt{conj/1}, $F$ is local variable}}\\
    \Longrightarrow&
    \begin{array}{c}
      \leftarrow (G_1,A_2,\ldots,A_n,\true) ; (\alpha = \delta,D_1,\ldots,D_n) ; (\alpha,\beta,B_1,\ldots,B_m)\\
      \vdots\\
      \square\\
      \text{(with solution $\theta$ s.t. $\alpha\theta = \nu$)}
    \end{array}
  \end{align*}
  Note that if $\sigma$ refutes
  $\leftarrow \alpha = \delta,D_1,\ldots,D_n$; then
  $\delta\sigma = \alpha\sigma$, hence~\eqref{eq:copy-term}
  $(D_1,\ldots,D_n)\sigma = (G_2,A_2,\ldots,A_n)\sigma$.
  Thus, $\sigma$ also refutes $\leftarrow G_2,A_2,\ldots,A_n$.

  So, we may write the above as:
  \begin{align*}
    &
    \begin{array}{c}
      \leftarrow (G_1,A_2,\ldots,A_n,\true) ; (G_2,A_2,\ldots,A_n) ; (\alpha = \beta,B_1,\ldots,B_m)\\
      \vdots\\
      \square\\
      \text{(with solution $\theta$ s.t. $\alpha\theta = \nu$)}
    \end{array}\\
    &\fbox{\textbf{Distributivity}}\\
    \Longrightarrow&
    \begin{array}{c}
      \leftarrow ((G_1;G_2),A_2,\ldots,A_n,\true) ; (\alpha = \beta,B_1,\ldots,B_m)\\
      \vdots\\
      \square\\
      \text{(with solution $\theta$ s.t. $\alpha\theta = \nu$)}
    \end{array}
  \end{align*}

  \pcase{[conj(Cs)|\Conj]} 
  \begin{align*}
    &
    \begin{array}{l}
      \qry \mathit{eval}([\conj([D_1,\ldots,D_l]),A_2,\ldots,A_n],\alpha,\mathit{alt}(\beta,(B_1,\ldots,B_m)),P,R).\\
      P = \nu,
      R = \mathit{success(\gamma,(C_1,\ldots,C_k))}.
    \end{array}\\
    &\fbox{\textbf{Definition and Determinism of \texttt{eval/5}}}\\
    \iff&
    \begin{array}{l}
      \qry \mathit{eval}([D_1,\ldots,D_l,A_2,\ldots,A_n],\alpha,\mathit{alt}(\beta,(B_1,\ldots,B_m)),P,R).\\
      P = \nu,
      R = \mathit{success(\gamma,(C_1,\ldots,C_k))}.
    \end{array}\\
    &\fbox{\textbf{Induction}}\\
    \Longrightarrow&
    \begin{array}{c}
      \leftarrow (D_1,\ldots,D_l,A_2,\ldots,A_n,\true) ; (\alpha = \beta,B_1,\ldots,B_m)\\
      \vdots\\
      \square\\
      \text{(with solution $\theta$ s.t. $\alpha\theta = \nu$)}
    \end{array}\\
    &\fbox{\textbf{Induction}}\\
    \Longrightarrow&
    \begin{array}{c}
      \leftarrow (\conj([D_1,\ldots,D_l]),A_2,\ldots,A_n,\true) ; (\alpha = \beta,B_1,\ldots,B_m)\\
      \vdots\\
      \square\\
      \text{(with solution $\theta$ s.t. $\alpha\theta = \nu$)}
    \end{array}
  \end{align*}
  
  \pcase{[C|\Conj]}
  In this case there are two possibilities: if no clauses match with the head,
  \texttt{Clauses} is empty, and the procedure backtracks. In this case the
  reasoning is identical to the case \texttt{[fail|\Conj]}.
  
  Otherwise, assume we have a list
  $[(H :\!\!- \mathit{Body})_1,\ldots,(H :\!\!-\mathit{Body})_l]$ of clauses
  that match $C$.
  After \texttt{disjoin\_clauses/3}, this list becomes a disjunction
  $(C = H_1,Body_1) ; \cdots ; (C = H_l,Body_l)$.

  From the recursive call we get:
  \begin{align*}
    &
    \begin{array}{l}
      \qry \mathit{eval}([C|A_2,\cdots,A_n],\alpha,\mathit{alt}(\beta,(B_1,\ldots,B_m)),P,R).\\
      P = \nu,
      R = result(\gamma,(C_1,\ldots,C_k))
    \end{array}\\
    &\fbox{\textbf{Definition and Determinism of \texttt{eval/5}}}\\
    \iff&
    \begin{array}{l}
      \qry \mathit{eval}([(C = H_1,Body_1) ; \cdots ; (C = H_l,Body_l)|A_2,\cdots,A_n],\alpha,\mathit{alt}(\beta(B_1,\ldots,B_m)),P,R).\\
      P = \nu,
      R = result(\gamma(C_1,\ldots,C_k))
    \end{array}\\
    &\fbox{\textbf{Induction}}\\
    \Longrightarrow&
    \begin{array}{c}
      \leftarrow ((C = H_1,Body_1) ; \cdots ; (C = H_l,Body_l)),A_2,\ldots,A_n,true) ; (\alpha = \beta, B_1,\ldots,B_m)\\
      \vdots\\
      \square\\
      \text{(with solution $\theta$, s.t. $\alpha\theta = \nu$)}
    \end{array}
  \end{align*}

  If $\leftarrow ((C = H_1,Body_1) ; \cdots ; (C = H_l,Body_l)),A_2,\ldots,A_n,true)$
  has a refutation $\theta$, then let $\theta = \theta_i\theta'$ such that
  $\theta_i$ refutes $\leftarrow C = H_i,Body_i$ and $\theta'$ refutes
  $\leftarrow A_2,\ldots,A_n$:
  \begin{align*}
  \begin{array}{c}\leftarrow C = H_i,Body_i\\\vdots\\\square\\\text{(with solution $\theta_i$)}\end{array}
  \Longrightarrow&
  \begin{array}{c}\leftarrow Body_i\sigma\\\vdots\\\square\\\text{(with solution $\theta_i'$)}\end{array}
  \text{s.t. $C\sigma = H_i\sigma$}
  \Longrightarrow
  \begin{array}{c}
    \leftarrow C\\
    \leftarrow Body_i\sigma\\
    \vdots\\
    \square\\
    \text{(with solution $\theta_i = \sigma\theta_i'$)}
  \end{array}
  \end{align*}
  Then $\theta$ also refutes $\leftarrow C,A_2,\ldots,A_n$, so we can conclude:
  \[\begin{array}{c}
    \leftarrow (C,A_2,\ldots,A_n,true) ; (\alpha = \beta, B_1,\ldots,B_m)\\
    \vdots\\
    \square\\
    \text{(with solution $\theta$, s.t. $\alpha\theta = \nu$)}
  \end{array}\]
\end{proof}

\newpage

\subsection{Evaluation Is Complete}
\label{sec:completeness}
\begin{AppendixTheorem}
  For all goals $\leftarrow A_1,\ldots,A_n$, if
  \[
  \begin{array}{c}
    \leftarrow A_1,\ldots,A_n\\
    \vdots\\
    \square\\
    \text{(with solution $\theta$)}
  \end{array}\]
  then
  \[\begin{array}{l}
  \qryP \mathit{eval}([A_1,\ldots,A_n],\alpha,\mathit{alt}(\beta,(B_1,\ldots,B_m)),P,R).\\
  P = \mathit{success}(\gamma,(C_1,\ldots,C_k)),R = \alpha\theta.
  \end{array}\]
  where $\qryP$ is defined as $\qry$, but each step is allowed to use a
  different permutation of the program.
\end{AppendixTheorem}
\begin{proof}
  By induction on the length of the derivation.

  \paragraph{Base} \fbox{$n = 0$}:
  $\leftarrow A_1,\ldots,A_n = \leftarrow \true = \square$ then $\theta$ is
  empty, and
  \[
  \begin{array}{l}
    \qryP \mathit{eval}([],\alpha,\mathit{alt}(\beta,(B_1,\ldots,B_m)),P,R).\\
    P = \mathit{success}(\beta,(B_1,\ldots,B_m)), R = \alpha
  \end{array}
  \]
  by definition.

  \paragraph{Induction} Let
  \[
  \begin{array}{cl}
    \leftarrow A_1,A_2,\ldots,A_n\\
    \leftarrow G_1,\ldots,G_f,A_2,\ldots,A_n\theta_1 &\text{(clause $H \leftarrow G_1,\ldots,G_f$; $A_1\theta_1 = H\theta_1$)}\\
    \vdots\\
    \square\\
    \text{(with substitution $\theta = \theta_1\theta_2$)}
  \end{array}
  \]
  By induction we have,
  \begin{align*}
    &\begin{array}{l}
      \qryP \mathit{eval}([G_1,\ldots,G_f,A_2,\ldots,A_n\theta_1],\alpha\theta_1,\mathit{alt}(\delta,(D_1,\ldots,D_l)),P,R)\\
      P = \mathit{success}(\gamma,(C_1,\ldots,C_k)), S = \alpha\theta_1\theta_2.
    \end{array}\\
    \iff&
    \begin{array}{l}
      \qryP \mathit{eval}([(A_1 = H,G_1,\ldots,G_f ; \ldots) ,A_2,\ldots,A_n],\alpha,\mathit{alt}(\beta,(B_1,\ldots,B_m)),P,R)\\
      P = \mathit{success}(\gamma,(C_1,\ldots,C_k)), S = \alpha\theta_1\theta_2.
    \end{array}\\
    \Longrightarrow&
    \begin{array}{l}
      \qryP \mathit{eval}([A_1,A_2,\ldots,A_n],\alpha,\mathit{alt}(\beta,(B_1,\ldots,B_m)),P,R)\\
      P = \mathit{success}(\gamma,(C_1,\ldots,C_k)), S = \alpha\theta_1\theta_2.
    \end{array}
  \end{align*}
\end{proof}
\newpage

\section{Additional Examples}

\subsection{Negation}
\begin{Verbatim}
not(G) :-
  copy_term(G,GC),
  reset(GC,GC,R),
  R = failure.
\end{Verbatim}

\subsection{Interoperable Engines}
\label{sec:app-interop-engines}
\begin{Verbatim}
engines(G) :-
    engines(G,[]).

new_engine(Pattern,Goal,Interactor) :-
    shift(new_engine(Pattern,Goal,Interactor)).

get(Interactor,Answer) :-
    shift(get(Interactor,Answer)).

return(X) :-
    shift(return(X)).

engines(G,EngineList) :-
    copy_term(G,GC),
    reset(GC,GC,R),
    engines_result(G,GC,EngineList,R).

engines_result(_,_,_,failure) :-
    fail.
engines_result(G,GC,EngineList,success(BC,B)) :-
    (G = GC ; G = BC, engines(B,EngineList)).
engines_result(G,GC,EngineList,S) :-
    S = shift(new_engine(Pattern,Goal,Interactor),C,BC,B),
    length(EngineList,Interactor),
    copy_term(Pattern-Goal,PatternCopy-GoalCopy),
    NewEngineList = [Interactor-engine(PatternCopy,GoalCopy)|EngineList],
    G = GC,
    G = BC,
    engines((C;B),NewEngineList).
engines_result(G,GC,EngineList,S) :-
    S = shift(get(Interactor,Answer),C,BC,B),
    member(Interactor-Engine,EngineList),
    run_engine(Engine,NewEngine,Answer),
    update(Interactor,NewEngine,EngineList,NewEngineList),
    G = GC,
    G = BC,
    engines((C;B),NewEngineList).

update(K,NewV,[K-_|T],[K-NewV|T]).
update(K,NewV,[OtherK-V|T],[OtherK-V|T2]) :-
    K \== OtherK,
    update(K,NewV,T,T2).
		  
run_engine(engine(Pattern,Goal),NewEngine,Answer) :-
    reset(Pattern,Goal,Result),
    run_engine_result(Pattern,NewEngine,Answer,Result).

run_engine_result(Pattern,NewEngine,Answer,failure) :-
    NewEngine = engine(Pattern,fail),
    Answer    = no.
run_engine_result(Pattern,NewEngine,Answer,success(BPattern,B)) :-
    NewEngine = engine(BPattern,B),
    Answer    = the(Pattern).
run_engine_result(Pattern,NewEngine,Answer,shift(return(X),C,BPattern,B)) :-
    BPattern  = Pattern,
    NewEngine = engine(Pattern,(C;B)),
    Answer    = the(X).
\end{Verbatim}

\subsection{ProbLog}
\label{sec:problog}

\begin{Verbatim}
fact(F) :-
    shift(fact(F,V)),
    V = t.

is_true(F,Pc) :- member(F-t,Pc).
is_false(F,Pc) :-member(F-f,Pc).

problog(Goal) :- problog(Goal,[]).
problog(Goal,Pc) :-
    reset(Goal,Goal,Result),
    analyze_problog(Result,Pc).

analyze_problog(success(_,_),_Pc).
analyze_problog(shift(fact(F,V),C,_,Branch),Pc) :-
    is_true(F,Pc),
    V = t,
    problog((C;Branch),Pc).
analyze_problog(shift(fact(F,V),C,_,Branch),Pc) :-
    is_false(F,Pc),
    V = f,
    problog((C;Branch),Pc).
analyze_problog(shift(fact(F,V),C,_,Branch),Pc) :-
    not(is_true(F,Pc)),
    not(is_false(F,Pc)),
    msw(F,V),
    problog((C;Branch),[F-V|Pc]).
analyze_problog(failure,_Pc) :- fail.
% NEGATION DOES NOT WORK.
\end{Verbatim}

Example usage:
\begin{Verbatim}
  % 0.5 :: f1.
  values_x(f1,[t,f],[0.5,0.5]).
  f1 :- fact(f1).
  % 0.5 :: f2.
  values_x(f2,[t,f],[0.5,0.5]).
  f2 :- fact(f2).

  p :- f1.
  p :- f2
  ?- solutions(prob(problog((f1,f1)))).
  problog((f1,f1)): 0.5
  ?- solutions(prob(problog(p))).
  problog(p): 0.75
\end{Verbatim}
\newpage

\section{Example Meta-Interpreter Trace}
\label{sec:example-trace}

Consider the simple program:
\begin{Verbatim}
  p(1).
  p(2) :- shift(2).
\end{Verbatim}

The following table shows the values of each of the arguments of the
meta-interpreter, while evaluating the goal \texttt{p(X)}.

\begin{tabular}{clc}
  \toprule
  \texttt{PatIn} & \texttt{Conj} & \texttt{Disj}\\
  \midrule
  \texttt{X} & \texttt{[p(X)]}                      & \texttt{alt(Y,fail)}\\
  \texttt{X} & \texttt{[(X=1,true;}                 & \texttt{alt(Y,fail)}\\
             & \texttt{X = 2,shift(2))]}            &\\
  \texttt{X} & \texttt{[(X=1,true)]}                & \texttt{alt(Z1,(Z1=X1,X1=2,shift(2) ; Z1=Y,fail))}\\
  \texttt{X} & \texttt{[X=1,true]}                  & \texttt{alt(Z1,(Z1=X1,X1=2,shift(2) ; Z1=Y,fail))}\\
  \texttt{1} & \texttt{[true]}                      & \texttt{alt(Z1,(Z1=X1,X1=2,shift(2) ; Z1=Y,fail))}\\
  \texttt{1} & \texttt{[]}                          & \texttt{alt(Z1,(Z1=X1,X1=2,shift(2) ; Z1=Y,fail))}\\
  \multicolumn{3}{l}{\texttt{PatOut=1}, \texttt{Result=success(Z1,(Z1=X1,X1=2,shift(2) ; Z1=Y,fail))}}\\
  \bottomrule
\end{tabular}\\

We can then evaluate the alternative branches:\\

\begin{tabular}{clc}
  \toprule
  \texttt{PatIn} & \texttt{Conj} & \texttt{Disj}\\
  \midrule
  \texttt{Z1} & \texttt{[(Z1=X1,X1=2,shift(2) ; Z1=Y,fail))]} & \texttt{alt(A,fail)}\\
  \texttt{Z1} & \texttt{[(Z1=X1,X1=2,shift(2))]} & \texttt{alt(B,(Z2=B,Z2=Y,fail;} \\
              &                                  & \texttt{A=B,fail)}\\
  \texttt{Z1} & \texttt{[Z1=X1,X1=2,shift(2)]}   & \texttt{alt(B,(Z2=B,Z2=Y,fail;} \\
              &                                  & \texttt{A=B,fail)}\\
  \texttt{X1} & \texttt{[X1=2,shift(2)]}         & \texttt{alt(B,(Z2=B,Z2=Y,fail;} \\
              &                                  & \texttt{A=B,fail)}\\
  \texttt{2}  & \texttt{[shift(2))}              & \texttt{alt(B,(Z2=B,Z2=Y,fail;} \\
              &                                  & \texttt{A=B,fail)}\\
  \texttt{2}  & \texttt{[]}                      & \texttt{alt(B,(Z2=B,Z2=Y,fail;} \\
              &                                  & \texttt{A=B,fail)}\\
  \multicolumn{3}{l}{\texttt{Result=shift(2,conj([]),B,(Z2=B,Z2=Y,fail;A=B,fail))}}\\
  \bottomrule
\end{tabular}\\

and again\\

\begin{tabular}{clc}
  \toprule\\
  B        & \texttt{[(Z2=B,Z2=Y,fail;A=B,fail)]} & \texttt{alt(C,fail)}\\
  $\vdots$ & $\vdots$                             &\vdots\\
  Y        & \texttt{[fail]} & \texttt{alt(D,B1=D,A=B1,fail ; C=D,fail)}\\
  \multicolumn{3}{l}{\texttt{backtracking}}\\
  D        & \texttt{[(B1=D,A=B1,fail ; C=D,fail)]} & \texttt{alt(E,fail)}\\
  $\vdots$ & $\vdots$                             &\vdots\\
  D        & \texttt{[fail]} & \texttt{alt(E,fail)}\\
  \multicolumn{3}{l}{backtracking fails: \texttt{Result=failure}}\\
  \bottomrule
\end{tabular}

\newpage

\section{Full Meta-Interpreter}
The following code has been tested on SWI-Prolog versions 7.6.3 and
8.0.3.
It uses the \texttt{type\_check} package to type-check the Prolog code.
The type-checking annotations can be commented out, without loss of
functionality on systems that do not support type-checking.

\begin{Verbatim}
%%-----------------------------------------------------------------------------
%% Meta-interpreter for disjunctive delimited control.
%%
%% The meta-interpreter supports most basic Prolog features, except if then
%% else.
%%
%% Author: Alexander Vandenbroucke (alexander.vandenbroucke@kuleuven.be)
%% Author: Tom Schrijvers (tom.schrijvers@kuleuven.be)
%%
%%-----------------------------------------------------------------------------

:- style_check(-singleton).
:- use_module(library(type_check)).


%% preliminaries

% The operator :=/2 is a more strictly typed synonym of =/2.
:- op(700,xfx,:=).
:- pred :=(A,A).
:=(X,Y) :- X = Y.

:- trust_pred copy_term(A,A).
:- trust_pred clause(pred,pred).
:- trust_pred shift(any).
:- trust_pred reset(P,pred,result(P)).
:- trust_pred write(any).
:- trust_pred length(list(A),integer).
:- trust_pred member(A,list(A)).



%%-----------------------------------------------------------------------------
%% all solutions
:- pred solutions(pred).
solutions(G0) :-
    copy_term(G0,G),
    copy_term(G,GC),
    empty_disj(EmptyDisj),
    do_reset((GC,true),G,GC,Result,EmptyDisj),
    analyze_solutions(G0,G,Result).

:- pred analyze_solutions(pred,pred,result(pred)).
analyze_solutions(G0,G,success(PatternCopy,Branch)) :-
    ( G0 = G
    ; copy_term(PatternCopy,PC),
      empty_disj(EmptyDisj),
      do_reset((Branch,true),PC,PatternCopy,Result,EmptyDisj),
      analyze_solutions(G0,PC,Result)
    ).
analyze_solutions(_,_,shift(_,_,_,_)) :-
    write('solutions: unexpected shift/1.\n'),
    fail.
analyze_solutions(_,_,failure) :- fail.



%%-----------------------------------------------------------------------------
%% Meta Interpreter

:- pred do_reset(pred,P,P,result(P),disjunction(P)).
%% do_reset(Conjuction,Pattern,PatternCopy,Result,Disjunction)
%%
%% * On success or shift, Pattern contains the current solution.
%%   Pattern MUST be an unbound variable.
%% * PatternCopy is a fresh copy of Pattern. Note: the contents of PatternCopy
%%   are not always sensible after do_reset/5, hence PatternCopy should not be
%%   inspected afterwards.
%% * Afterwards Result is one of:
%%    - success(BranchPattern,Branch): BranchPattern is a fresh copy of
%%      Pattern, Branch is a disjunctive continuation (see below). In this case
%%      PatternCopy and Pattern are unified, to carry out the current
%%      substitution.
%%    - shift(Term,Conjunction,BranchPattern,Branch):
%%      Term is the term that was shifted, Conjuction is the conjunctive
%%      continuation. BranchPattern is a fresh copy of Pattern, Branch is the
%%      disjunctive continuation. Pattern and PatternCopy are unified to carry
%%      out the current substitution.
%%    - failure: when the conjunction and the disjunction fail
%%   Result MUST be an unbound variable.
%% * Conjunction: the current conjunctive goal. The variables it contains
%%   should not occur in Pattern. Moreover, the final conjunct should always be
%%   true/0. The conjunction must NOT be empty.
%% * Disjunction: a disjunction/2, that contains the disjunctive goal (see
%%   empty_disj/1, disjoin/3 below.

%% single true/0-pattern (base case)
do_reset(true,Pattern,PatternCopy,Result,Disj) :-
    !,
    disjunction(BranchPattern,Branch) := Disj,
    Pattern := PatternCopy,
    Result  := success(BranchPattern,Branch).

%% fail/0-pattern
do_reset((fail,_Conj),Pattern,_,Result,Disj) :-
    !,
    backtrack(Pattern,Result,Disj).

%% true/0-pattern
do_reset((true,Conj),Pattern,PatternCopy,Result,Disj) :-
    !,
    do_reset(Conj,Pattern,PatternCopy,Result,Disj).

%% conjunction pattern
do_reset(((G1,G2),Conj),Pattern,PatternCopy,Result,Disj) :-
    !,
    do_reset((G1,(G2,Conj)),Pattern,PatternCopy,Result,Disj).

%% disjunction pattern
do_reset(((G1;G2),Conj),Pattern,PatternCopy,Result,Disj) :-
    !,
    copy_term(disjunction(PatternCopy,(G2,Conj)),Branch),
    disjoin(Branch,Disj,NewDisj),
    do_reset((G1,Conj),Pattern,PatternCopy,Result,NewDisj).

%% shift/1-pattern
do_reset((shift(X),Conj),Pattern,PatternCopy,Result,Disj) :-
    !,
    Pattern := PatternCopy,
    Disj    := disjunction(BranchPattern,Branch),
    Result  := shift(X,Conj,BranchPattern,Branch).

%% (new) reset/3-pattern
do_reset((reset(P,G,R),Conj),Pattern,PatternCopy,Result,Disj) :-
    !,
    copy_term(P-G,PC-GC),
    empty_disj(D),
    do_reset((GC,true),Q,PC,S,D),
    type_to_any(R,RAny),type_to_any(S,SAny),
    do_reset((P = Q, (RAny = SAny, Conj)),Pattern,PatternCopy,Result,Disj).

%% unification pattern
do_reset((X = Y,Conj),Pattern,PatternCopy,Result,Disj) :-
    !,
    ( X = Y -> do_reset(Conj,Pattern,PatternCopy,Result,Disj)
    ; backtrack(Pattern,Result,Disj)).

%% unification pattern
do_reset((X = Y,Conj),Pattern,PatternCopy,Result,Disj) :-
    !,
    ( X = Y -> do_reset(Conj,Pattern,PatternCopy,Result,Disj)
    ; backtrack(Pattern,Result,Disj)).

%%term non-equivalence pattern
do_reset((X \== Y,Conj),Pattern,PatternCopy,Result,Disj) :-
    !,
    ( X \== Y -> do_reset(Conj,Pattern,PatternCopy,Result,Disj)
    ; backtrack(Pattern,Result,Disj)).

%% copy_term/2-pattern
do_reset((copy_term(X,Y),Conj),Pattern,PatternCopy,Result,Disj) :-
    !,
    ( copy_term(X,Y) -> do_reset(Conj,Pattern,PatternCopy,Result,Disj)
    ; backtrack(Pattern,Result,Disj)).

%% cut/0-pattern
do_reset((!,Conj),Pattern,PatternCopy,Result,Disj) :-
    !,
    do_reset(Conj,Pattern,PatternCopy,Result,Disj).

%% findall/3-pattern
do_reset((findall(T,G,L),Conj),Pattern,PatternCopy,Result,Disj) :-
    !,
    ( findall(T,G,L) -> do_reset(Conj,Pattern,PatternCopy,Result,Disj)
    ; backtrack(Pattern,Result,Disj)).

%% length/2-pattern
do_reset((length(L,X),Conj),Pattern,PatternCopy,Result,Disj) :-
    !,
    ( length(L,X) -> do_reset(Conj,Pattern,PatternCopy,Result,Disj)
    ; backtrack(Pattern,Result,Disj)).

%% member/2-pattern
do_reset((member(X,L),Conj),Pattern,PatternCopy,Result,Disj) :-
    !,
    ( member(X,L) -> do_reset(Conj,Pattern,PatternCopy,Result,Disj)
    ; backtrack(Pattern,Result,Disj)).

%% is/2-pattern
do_reset((is(R,Exp),Conj),Pattern,PatternCopy,Result,Disj) :-
    !,
    ( R is Exp -> do_reset(Conj,Pattern,PatternCopy,Result,Disj)
    ; backtrack(Pattern,Result,Disj)).

%% (<)/2-pattern
do_reset((<(X,Y),Conj),Pattern,PatternCopy,Result,Disj) :-
    !,
    ( X < Y -> do_reset(Conj,Pattern,PatternCopy,Result,Disj)
    ; backtrack(Pattern,Result,Disj)).

%% (>=)/2-pattern
do_reset((>=(X,Y),Conj),Pattern,PatternCopy,Result,Disj) :-
    !,
    ( X >= Y -> do_reset(Conj,Pattern,PatternCopy,Result,Disj)
    ; backtrack(Pattern,Result,Disj)).

%% (@<)/2-pattern
do_reset((@<(X,Y),Conj),Pattern,PatternCopy,Result,Disj) :-
    !,
    ( X @< Y -> do_reset(Conj,Pattern,PatternCopy,Result,Disj)
    ; backtrack(Pattern,Result,Disj)).

%% (@>=)/2-pattern
do_reset((@>=(X,Y),Conj),Pattern,PatternCopy,Result,Disj) :-
    !,
    ( X @>= Y -> do_reset(Conj,Pattern,PatternCopy,Result,Disj)
    ; backtrack(Pattern,Result,Disj)).

%% write/1-pattern
do_reset((write(T),Conj),Pattern,PatternCopy,Result,Disj) :-
    !,
    write(T),
    do_reset(Conj,Pattern,PatternCopy,Result,Disj).

%% clause pattern.
do_reset((G,Conj),Pattern,PatternCopy,Result,Disj) :-
    !,
    findall(GC-Body,(clause(G,Body), GC = G),Clauses),
    ( Clauses := [] ->
        backtrack(Pattern,Result,Disj)
    ; disjoin_clauses(G,Clauses,ClausesDisj),
      do_reset((ClausesDisj,Conj),Pattern,PatternCopy,Result,Disj)
    ).

:- pred disjoin_clauses(P,list(pair(P,pred)),pred).
disjoin_clauses(_,[],fail).
% The next clause is not necessary, but makes things prettier when tracing.
disjoin_clauses(G,[GC-Clause],(G=GC,Clause)) :- !.
disjoin_clauses(G,[GC-Clause|Clauses], ((G=GC,Clause) ; Disj) ) :-
    disjoin_clauses(G,Clauses,Disj).

:- pred backtrack(P,result(P),disjunction(P)).
%% backtracking
backtrack(Pattern,Result,Disj) :-
    ( empty_disj(Disj) ->
         Result := failure
    ; Disj := disjunction(PatternCopy,G) ->
         empty_disj(EmptyDisj),
	 do_reset((G,true),Pattern,PatternCopy,Result,EmptyDisj)
         %% TODO: could be more efficient if we pattern match on G?
    ).


%%-----------------------------------------------------------------------------
%% Disjunctions
%%
%% Disjunctions are of the from disjunction(Pattern,Goal), where Goal may
%% contain a disjunction.

%% Disjunctions are neither commutative nor associative. However empty_disj/1
%% is still the unit for disjoin/3.

:- type disjunction(Pattern) ---> disjunction(Pattern,pred).

%% The empty disjunction.
:- pred empty_disj(disjunction(Pattern)).
empty_disj(disjunction(_,fail)).

%% Disjoin two disjunctions.
%%
%% This operation is not commutative
:- pred disjoin(disjunction(P),disjunction(P),disjunction(P)).
disjoin(disjunction(_,fail),Disjunction,Disjunction) :- !.
disjoin(Disjunction,disjunction(_,fail),Disjunction) :- !.
disjoin(disjunction(PC1,G1),disjunction(PC2,G2),Disjunction) :-
    PC1 := PC2,
    Disjunction := disjunction(PC1, (G1 ; G2)).
% NOTE: things could probably be made more efficient by not unifying PC1
% and PC2, and keeping the disjunctions explicit, s.t. no fresh copy of
% G2 needs to be made, but this is way simpler.


%%-----------------------------------------------------------------------------
%% Results
:- type result(R) ---> success(R,pred) ; shift(any,pred,R,pred) ; failure.

\end{Verbatim}

\label{lastpage}
\end{document}